\documentclass[journal]{IEEEtran}

\usepackage[english]{babel}
\usepackage[utf8]{inputenc}
\usepackage{amsmath,amssymb,bm,mathtools}
\usepackage{amsthm}
\usepackage{booktabs,multirow}
\usepackage{graphicx}
\usepackage{xcolor}
\usepackage{xparse}
\usepackage{cite}
\usepackage{microtype}
\let\newblock\relax
\NewDocumentCommand{\newblock}{+g}{\IfNoValueTF{#1}{\relax}{#1}}
\NewDocumentCommand{\vtwo}{+g}{\IfNoValueTF{#1}{\relax}{#1}}
\graphicspath{{figures/}{./}}

\newtheorem{theorem}{Theorem}
\newtheorem{corollary}{Corollary}
\newtheorem{remark}{Remark}
\begin{document}

\title{Branch-Level Energy Localization in Three-Phase Loads: Resolving Indeterminacy in Time-Domain}

\author{Francisco G. Montoya, Francisco de Leon, \IEEEmembership{Life Fellow, IEEE},
Francisco M. Arrabal-Campos, and Alfredo Alcayde%
\thanks{F. G. Montoya, F. M. Arrabal-Campos, and A. Alcayde are with the Department of Engineering, University of Almeria, Spain. F. de Leon is with the Department of Electrical and Computer Engineering, New York University, USA.}}

\maketitle

\begin{abstract}
This paper develops a branch-level energy-localization framework for three-phase loads. The instantaneous terminal power of an admissible lumped equivalent is decomposed uniquely as Joule dissipation plus magnetic and electric stored-energy rates, branch by branch. Three formal results are established: a Branch-Level Localization Theorem (uniqueness given an admissible topology); a Topology-Indeterminacy Theorem (multiple admissible topologies reproduce identical terminal data with distinct localizations); and a Generalized Energetic Duality Theorem that organizes classical electrical dualities (Norton-Thevenin, series--parallel, $L$ vs $C$, $R$ vs $G$) as restrictions to Linear Time Invariant (LTI) sinusoidal regimes of a single time-domain principle in which constant-parameter equivalence is replaced by time-varying parameters. The framework is exercised on six test cases including the de Leon--Cohen open-phase paradox, switched-resistive loads, three-wire delta-versus-wye-virtual indeterminacy, fluctuating-phase loads, and a four-wire nonlinear load with hysteretic, linear, and switched branches. The framework is positioned as complementary to IEEE Std.~1459, CPC, instantaneous $p$--$q$, and Fryze-Buchholz-Depenbrock: each answers a different question, and the apparent paradoxes vanish once the question is posed precisely.
\end{abstract}


\section{Introduction}

\IEEEPARstart{P}{ower} theory has a long-standing tension between two questions \cite{emanuel2010power,IEEE1459-2025}. The first is operational: which numerical quantities, computed from voltage and current waveforms at a measurement section, are useful for metering, equipment rating, compensation design, protection, and converter control? The second is physical: which of those quantities correspond to actual energy flow, irreversible dissipation, or the time rate of stored electromagnetic energy in the lumped circuit that the section terminates? The first question has been answered, in different ways, by Steinmetz \cite{steinmetz1916theory}, Budeanu, Fryze, Shepherd--Zakikhani, Kusters--Moore \cite{kusters1980determination}, Czarnecki's Currents' Physical Components \cite{czarnecki2008currents}, the Fryze--Buchholz--Depenbrock framework \cite{depenbrock1980active,staudt2008fryze}, the instantaneous $p$--$q$ theory of Akagi et al.~\cite{akagi1983generalized,akagi1984instantaneous,akagi2017instantaneous}, or IEEE Std.~1459 \cite{IEEE1459-2025}. The second question is harder; it has produced a sequence of paradoxes \cite{cakareski1999physical,emanuel2004poynting,czarnecki1997energy,garcia2007power,willems2010budeanu,calamaro2015review} where mathematically valid waveform decompositions yield mutually incompatible interpretations of where energy is dissipated, stored, or exchanged.

A recent analysis by some of the present authors \cite{poyatos2025comparative} examined four well-known single-phase paradoxes through the lens of Budeanu, Fryze, CPC, CPT, and the Maxwell-Poynting-based interpretation of de Leon and Cohen \cite{deleon2010poynting,de2005discussion,de2008discussion}, and concluded that classical theories often disagree on the same data. The present paper takes the next step and addresses three-phase systems, where the difficulty is structurally larger: in three-wire connections the neutral point is not accessible, line currents are linear combinations of internal branch currents, and several admissible lumped topologies can reproduce the same terminal observables. Aggregation across phases can also conceal local exchange, so that per-phase storage variations vanish at the three-phase boundary while remaining physically present.

\newblock{The thesis of this paper is that classical power-theory paradoxes in three-phase systems are not failures of physics or of any specific theory; they are symptoms of describing the system through terminal-only coordinates. Once a physically admissible lumped equivalent has been identified, the instantaneous terminal power admits a unique branch-level decomposition into Joule dissipation and stored-energy rates. The paradoxes then disappear by exhibition: the framework returns the location, sign, and time profile of every storage and dissipation contribution.}

The paper makes the following contributions. First, an energy-balance criterion for three-phase measured sections is formalized as a circuit-level statement, written in the quasi-static lumped regime where Poynting's theorem reduces to Tellegen's theorem \cite{tellegen1952general} together with constitutive relations. \newblock{No claim is made that reactive or apparent power follows from Maxwell's equations in any direct sense; field-theoretic localization is the proper subject of Poynting-based analyses \cite{emanuel2004poynting,calamaro2015review}, which are complementary to the present circuit-level treatment.} Second, three formal results are stated and proven: the Branch-Level Localization Theorem (uniqueness given an admissible topology), the Topology-Indeterminacy Theorem (multiple admissible topologies reproducing identical terminal data with distinct localizations), and the Generalized Energetic Duality Theorem (for any branch terminal pair, every admissible topology family can be parameterized to reproduce the same branch energy decomposition; classical Norton-Thevenin, series-parallel, and $L\leftrightarrow C$ dualities are special cases of this principle in which the parameter trajectories degenerate to constants). Third, the framework is exercised on six test cases, including an explicit indeterminacy construction not previously documented, a switched-resistive nonlinear load that exposes the coordinate-artifact nature of classical distortion, scattered, void, and instantaneous reactive components when no storage element is present, and a four-wire nonlinear case combining hysteresis, linear $RLC$, and switched branches that documents the framework on non-three-wire and structurally heterogeneous loads. Fourth, a comparative table evaluates how IEEE Std.~1459, CPC, instantaneous $p$--$q$, and FBD describe these cases relative to the branch-level interpretation; the framework is positioned as complementary, not as a replacement.

The lumped equivalents used to perform the localization are obtained from the time-domain identification machinery developed in companion work for single-phase \cite{montoyaTPWRD1,arrabal2025experimental} and three-phase \cite{montoya3phaseID} systems; the identification methodology is invoked here only as enabling technology and is not redeveloped. The paper is organized as follows. Section~\ref{sec:framework} establishes the energy-balance framework, proves the three theorems, and discusses the duality hierarchy with a short pedagogical illustration in single-phase. Section~\ref{sec:localization} specializes the framework to three-phase delta-parallel and wye-series equivalents and discusses the role of three-wire connections. Section~\ref{sec:paradoxes} presents the six test cases. Section~\ref{sec:comparison} compares the localization with classical theories. Sections~\ref{sec:discussion} and \ref{sec:conclusion} discuss limitations and conclude.

\section{Energy-Balance Framework}
\label{sec:framework}

\subsection{Working assumptions and notation}

Let $\Omega$ denote a three-phase measured section described as a connected lumped-element subnetwork. Instantaneous voltages and currents at the terminal ports are the only available observables. \vtwo{Unless otherwise stated, voltages and currents are functions of time, and their explicit argument $t$ is omitted for brevity.} Inside $\Omega$, the network is composed of resistors, inductors, and capacitors with constant or slowly varying parameters; nonlinear or genuinely time-varying behavior is treated within the same framework by allowing branch parameters to depend on time, as is consistent with the identification literature \cite{montoyaTPWRD1,arrabal2025experimental}; see Fig. \ref{fig:fig1}(a). The quasi-static lumped regime is assumed throughout: characteristic dimensions of $\Omega$ are small compared with the wavelength of every relevant frequency component, so that displacement-current effects external to capacitive elements are negligible and Tellegen's theorem applies \cite{tellegen1952general}.

The Poynting integral form, which sets the global energy balance for $\Omega$, reduces in this regime to a circuit-level statement, say Kirchhoff Laws . \newblock{This statement is used here as the working axiom and is not derived from Maxwell's equations within the paper. Field-theoretic accounts of polyphase energy flow are available \cite{emanuel2004poynting,calamaro2015review}; the present analysis operates one level higher, on lumped quantities, where Tellegen's theorem provides the conservation relation between branch and terminal powers.}

\subsection{Branch energy balance}

Let $\mathcal R$, $\mathcal L$, and $\mathcal C$ denote the sets of resistive, inductive, and capacitive branches of an admissible lumped representation of $\Omega$. With currents $i_r$, $i_\ell$ and voltages $v_c$ associated with the corresponding elements, define
\begin{equation}
p_R=\sum_{r\in\mathcal R} R_r i_r^2,
\label{eq:joule_total}
\end{equation}
\begin{equation}
W_L=\sum_{\ell\in\mathcal L}\tfrac12 L_\ell i_\ell^2,\quad
W_C=\sum_{c\in\mathcal C}\tfrac12 C_c v_c^2.
\label{eq:stored_energies}
\end{equation}
Tellegen's theorem applied to the lumped equivalent of $\Omega$ gives the instantaneous terminal balance
\begin{equation}
p_{\rm term}=\bm v^\top\bm i=p_R(t)+W_L'(t)+W_C'(t),
\label{eq:tellegen_balance}
\end{equation}
where $\bm v$ and $\bm i$ collect the port voltages and currents in vector form and the prime denotes time derivative.

\newblock{Equation~\eqref{eq:tellegen_balance} is not a definition of new power components and does not need one. It is the equality of two scalar functions of time: the inner product $\bm v^\top\bm i$ measured at the terminal ports, and the algebraic sum of branch contributions inside $\Omega$. Whenever the lumped equivalent reproduces the measurements, the equality holds pointwise; when it does not, the residual is itself the diagnostic.}

\newblock{For a single resistive, inductive, or capacitive branch, the branch contributions to \eqref{eq:tellegen_balance} are
\begin{equation}
p_{R,r}=R_r i_r^2=G_r v_r^2,
\label{eq:branch_joule}
\end{equation}
\begin{equation}
W_{L,\ell}'=i_\ell L_\ell  i_\ell'=i_\ell v_{L,\ell},
\label{eq:branch_inductor_rate}
\end{equation}
\begin{equation}
W_{C,c}'=v_c C_c v_c'=v_c i_{C,c},
\label{eq:branch_capacitor_rate}
\end{equation}
where $G_r=1/R_r$ and the alternative right-hand sides hold by the constitutive relations $v_r=R_r i_r$, $v_{L,\ell}=L_\ell i_\ell'$, and $i_{C,c}=C_c v_c'$. The branch-level decomposition of $p_{\rm term}$ is therefore the term-by-term sum of \eqref{eq:branch_joule}--\eqref{eq:branch_capacitor_rate}.}

\subsection{Operational definition of physical power terms}

Three classes of physical power terms are admitted: terminal flow $p_{\rm term}$, irreversible Joule dissipation $p_R$, and the magnetic and electric stored-energy rates $W_L'$ and $W_C'$. Any quantity that can be associated with an entry in \eqref{eq:tellegen_balance} for an identified equivalent of $\Omega$ is granted the status of physical term in the sense of this paper. Traditional quantities that cannot be so associated, including the apparent power $S$, the Budeanu reactive power $Q_B$, the Czarnecki CPC components, and the instantaneous Akagi $p$--$q$ coordinates, retain their value as engineering indices but are not treated as physical power terms in the same sense. They are not declared incorrect; the claim is only that they are objects of a different kind. Section~\ref{sec:comparison} examines this distinction case by case.

\subsection{Branch-Level Localization Theorem}

The first formal result asserts uniqueness of the localization given an admissible topology and identified parameters.

\begin{theorem}[Branch-Level Localization]
\label{thm:localization}
Let $\mathcal T$ be an admissible lumped topology for $\Omega$ with branch sets $\mathcal R\cup\mathcal L\cup\mathcal C$, and let $\bm\theta=\{R_r,L_\ell,C_c\}$ be a parameter vector consistent with the terminal observables in the sense that the lumped equivalent reproduces $\bm v,\bm i$ within a residual $\rho_E\le\tau$ over the analysis window. Then the decomposition
\begin{equation}
p_{\rm term}=\sum_{r\in\mathcal R} R_r i_r^2+\sum_{\ell\in\mathcal L} L_\ell i_\ell i_\ell'+\sum_{c\in\mathcal C} C_c v_cv_c'
\label{eq:branch_decomposition}
\end{equation}
is unique.
\end{theorem}

\begin{IEEEproof}
By Tellegen's theorem applied to $\mathcal T$, the algebraic sum of branch powers equals $p_{\rm term}$ pointwise once branch voltages and currents are determined. For a fixed admissible topology, Kirchhoff's current and voltage laws plus the constitutive relations $v_r=R_r i_r$, $v_\ell=L_\ell i_\ell'$, and $i_c=C_c v_c'$ uniquely determine internal branch variables $\{i_r,i_\ell,v_c\}$ as linear functions of port observables and identified parameters. Consequently, each summand of \eqref{eq:branch_decomposition} is a deterministic function of the input data, and the right-hand side is a unique decomposition of $p_{\rm term}$ for that topology.
\end{IEEEproof}

\begin{remark}
\label{rem:uniqueness}
Theorem~\ref{thm:localization} concerns uniqueness within a fixed topology $\mathcal T$. It does not assert uniqueness across topologies. A different admissible topology yields a different branch decomposition of the same $p_{\rm term}$. This is the subject of Theorem~\ref{thm:indeterminacy}.
\end{remark}

\subsection{Topology-Indeterminacy Theorem}

The second result establishes that several admissible topologies may produce identical terminal observables while assigning energy distinctly across branches.

\begin{theorem}[Topology Indeterminacy]
\label{thm:indeterminacy}
There exist pairs $(\mathcal T_1,\bm\theta_1)\neq(\mathcal T_2,\bm\theta_2)$ of admissible topologies and parameter vectors that yield identical terminal observables $\bm v,\bm i$ over a finite analysis window but distinct branch-level decompositions \eqref{eq:branch_decomposition}.
\end{theorem}

\begin{IEEEproof}
A constructive proof is given. Consider an unbalanced three-wire load whose terminal data $(v_{ab},v_{bc},i_a,i_b)$ coincide with the response of (i) a delta-parallel equivalent with branch admittances $\{G_{xy},\Gamma_{xy},C_{xy}\}$, and (ii) a wye-series equivalent referred to a virtual neutral, with branch impedances $\{R_x,L_x,S_x\}$. Both equivalents satisfy Kirchhoff laws and the corresponding constitutive relations. They differ in which internal currents and voltages they declare physical, hence in which branches they assign $p_R$, $W_L'$, and $W_C'$. Section~\ref{sec:paradoxes} provides a numerical instance of this construction in which the delta equivalent assigns zero storage rates whereas the wye-virtual equivalent assigns nonzero $W_L'$ to symmetrized branches.
\end{IEEEproof}

\begin{corollary}
\label{cor:topology}
Branch-level localization is a property of the model, not solely of the data. The selection of a topology is a modeling decision that must accompany any claim of physical interpretation.
\end{corollary}

\begin{corollary}[Series-parallel duality]
\label{cor:series_parallel}
Let $\mathcal B_p$ be a parallel branch with parameters $(G,\Gamma,C)$ and constant scalars, and let $\mathcal B_s$ be a series branch with parameters $(R(t),L(t),S(t))$. If both reproduce the same branch terminal pair $(v,i)$ over the analysis window, then the branch energy decomposition \eqref{eq:branch_joule}--\eqref{eq:branch_capacitor_rate} returns identical $p_R(t)$, $W_L'(t)$, and $W_C'(t)$ in both representations, even though the parameter values are not equal in general, where $\breve{v}$ and $\breve{\imath}$ denote the time primitives of $v$ and $i$:
\begin{equation}
R_s(t)i^2=G_p v^2,
\label{eq:dual_joule}
\end{equation}
\begin{equation}
L_s(t)\,ii'=\Gamma_p\,v\breve v,
\label{eq:dual_inductor}
\end{equation}
\begin{equation}
S_s(t)\,i\breve\imath=C_p\,vv'.
\label{eq:dual_capacitor}
\end{equation}
\end{corollary}

\begin{IEEEproof}
For a single branch, the constitutive relations of the two representations and Kirchhoff's laws yield two parameterizations of the same branch terminal pair. The branch contributions in \eqref{eq:branch_joule}--\eqref{eq:branch_capacitor_rate} are functions of $(v,i)$ alone once expressed by either side of those equations; substituting the alternative parameterizations gives \eqref{eq:dual_joule}--\eqref{eq:dual_capacitor}.
\end{IEEEproof}

\newblock{Corollary~\ref{cor:series_parallel} is operationally relevant: when a parallel identification has been performed, the equivalent series circuit can be obtained by setting $R_s=v^2 G_p/i^2$, $L_s=v\breve v\,\Gamma_p/(i i')$, and $S_s=v v'\,C_p/(i\breve\imath)$. The series parameters are time-varying even when the parallel parameters are constants; the branch energy decomposition is invariant.}

\subsection{Generalized Energetic Duality}
\label{subsec:duality}

\newblock{Classical electrical duality is a structural property of LTI circuits with single-frequency sinusoidal excitation. It pairs the Kirchhoff voltage law of a series subcircuit with the Kirchhoff current law of its parallel dual under the element-level swap $R\leftrightarrow G=1/R$, $L\leftrightarrow C$, voltage source $\leftrightarrow$ current source, and produces identical frequency-domain responses, identical instantaneous active and reactive powers, and identical stored energy under the appropriate parameter mapping. Outside its operating regime, the classical duality breaks in two ways: nonsinusoidal excitation forbids a single set of constant dual parameters that hold simultaneously at every harmonic, and nonlinear or time-varying loads invalidate constant-parameter equivalence over a finite analysis window \cite{willems2010budeanu,jeltsema2015fluctuating}.}

\newblock{Both failures vanish at the branch energy level. The branch energy decomposition $p_R+W_L'+W_C'$ is invariant under a richer duality that absorbs both extensions, but paying the price that parameters become time-varying.}

\begin{theorem}[Generalized Energetic Duality]
\label{thm:duality}
Let $(v,i)$ be a branch terminal pair over a finite analysis window. For every admissible topology family $\mathcal T\in\{\mathrm{series},\mathrm{parallel},\mathrm{mixed}\}$, there exists a parameter trajectory $\bm\theta_{\mathcal T}(t)$ such that the branch energy decomposition $p_R(t)+W_L'(t)+W_C'(t)$ computed with $(v,i,\bm\theta_{\mathcal T})$ is identical for every $\mathcal T$. The trajectories degenerate to constants if and only if $(v,i)$ are sinusoidal at a single frequency and the underlying load is LTI; otherwise, at least one parameter is genuinely time-varying in at least one topology.
\end{theorem}

\begin{IEEEproof}
The constructive map between parallel and series representations is provided by Corollary~\ref{cor:series_parallel}: setting $R_s(t)=v^2 G_p/i^2$, $L_s(t)=v\breve v\,\Gamma_p/(ii')$, and $S_s(t)=vv'\,C_p/(i\breve\imath)$ enforces $p_R^{(s)}=p_R^{(p)}$, $W_L^{\prime (s)}=W_L^{\prime (p)}$, and $W_C^{\prime (s)}=W_C^{\prime (p)}$ pointwise via \eqref{eq:dual_joule}--\eqref{eq:dual_capacitor}. Mixed topologies are constructed by partitioning the energy decomposition between series and parallel branches and applying the mapping branch by branch. The constancy claim follows from substituting $v=V\sin\omega t$ and $i=I\sin(\omega t-\phi)$ into the ratios $v^2/i^2$, $v\breve v\,/(i i')$, and $v v'\,/(i\breve\imath)$: all three reduce to time-invariant constants only at a single frequency with constant amplitude. Multi-frequency excitation breaks this collapse and the ratios become explicit functions of $t$.
\end{IEEEproof}

\newblock{Theorem~\ref{thm:duality} organizes classical duality results as restrictions of a single time-domain principle:}

\begin{itemize}
\item \newblock{\textit{DC steady state}: only resistive elements; the duality reduces to $R\leftrightarrow G=1/R$, with constant parameters.}
\item \newblock{\textit{LTI sinusoidal at a single frequency}: classical Norton-Thevenin and the element-level swap $L\leftrightarrow C$ at the operating frequency, with constant parameters that absorb both the magnitude and phase of the dual relationship.}
\item \newblock{\textit{Nonsinusoidal LTI}: classical element-level duality recovers the equivalence at each harmonic separately; an aggregated time-domain equivalence requires time-varying parameters in at least one topology \vtwo{of fixed order, or equivalently constant elements arranged in a higher-order Foster or Cauer network~\cite{foster1924reactance,cauer1926synthesis,guillemin1957synthesis}}.}
\item \newblock{\textit{Nonlinear or genuinely time-varying loads}: constant-parameter dual representations fail; energetic duality is recovered only with time-varying parameter trajectories in every topology.}
\end{itemize}

\vtwo{The Foster (series and parallel) and Cauer (ladder) canonical forms make this freedom explicit: a prescribed LTI frequency-dependent one-port is realized by constant $L$ and $C$ elements arranged in different higher-order networks, each reproducing the same terminal behavior while storing energy in distinct internal elements. They are constant-element instances of the topology indeterminacy of Theorem~\ref{thm:indeterminacy}, complementary to the low-order time-varying branch used here: choosing a higher-order constant-element realization or a low-order time-varying one is again a modeling decision that fixes where the stored energy resides.}

\newblock{The three-wire topology indeterminacy of Section~\ref{subsec:topology_indeterminacy} is the multi-port instance of the same principle. The terminal data $(v_{ab},v_{bc},i_a,i_b)$ is reproduced by a delta-parallel parameterization and by a wye-series-virtual parameterization; the parameter mapping between them is not the classical element-level swap but a topology transformation across families. Theorem~\ref{thm:duality} guarantees that the energetic content of the load is invariant under this transformation, while Corollary~\ref{cor:topology} reminds that the choice of family is a modeling decision that must be reported. The two theorems are complementary: Theorem~\ref{thm:indeterminacy} states that distinct localizations exist; Theorem~\ref{thm:duality} states that, with a constructive parameter map, those localizations can be brought into one-to-one correspondence at the energy level. A numerical validation of the duality on the linear parallel $RLC$ of phase $b$ of the four-wire nonlinear case is reported later in Section~\ref{subsec:fourwire_nonlinear}.}

\subsection{Pedagogical illustration: Czarnecki C\,$\leftrightarrow$\,L paradox}
\label{subsec:czarnecki_cl}

A minimal single-phase example, due to Czarnecki~\cite{czarnecki1997budeanu}, fixes the meaning of Theorem~\ref{thm:indeterminacy}; see Fig. \ref{fig:fig1}(b). A parallel branch with $C=49/32$~F and $L=32/57$~H is excited by $v=\sqrt2\,(100\sin\omega t+25\sin 3\omega t)$ at $\omega=1$~rad/s. Czarnecki shows a different two-element circuit that produces an identical terminal current at the fundamental component. Budeanu's reactive power is $Q_B=0$.

\begin{figure}[!t]
\centering
\includegraphics[width=\columnwidth]{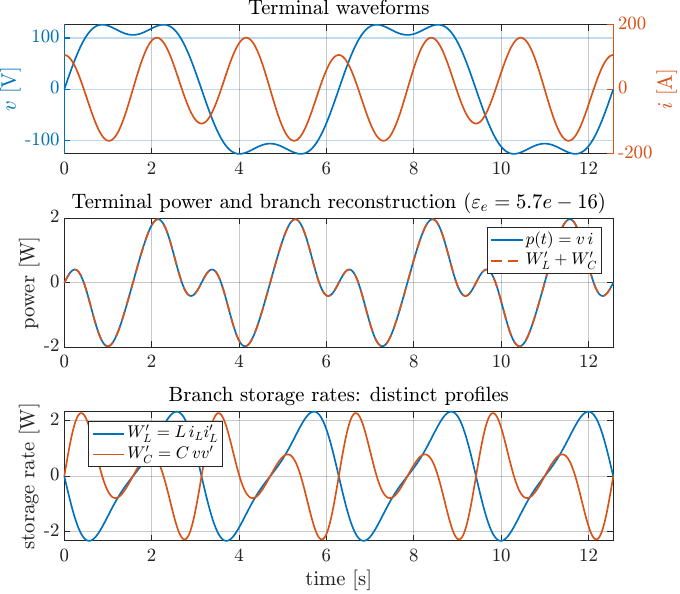}
\caption{\newblock{Czarnecki C\,$\leftrightarrow$\,L paradox as illustration of Theorem~\ref{thm:indeterminacy}. Top: terminal voltage and current (normalized, $\omega=1$ rad/s). Middle: instantaneous terminal power and the branch decomposition $W_L'+W_C'$, which reconstructs $p(t)$ exactly. Bottom: $W_L'$ and $W_C'$ are distinct functions of time, even though Budeanu's $Q_B=0$ and the fundamental components alone leave $L$ and $C$ indistinguishable.}}
\label{fig:p2_czarnecki_CL}
\end{figure}

The branch decomposition gives $W_L'=L\,i_L i_L'$ and $W_C'=C\,v v'$ with $W_L'+W_C'=p_{\rm term}$, since no resistor is present. As Fig.~\ref{fig:p2_czarnecki_CL} shows, $W_L'$ and $W_C'$ are distinct functions of time; their sum reconstructs the measured $p(t)$. The terminal data alone do not distinguish circuits whose fundamental currents agree; the energy localization does. The same construction generalizes to three-phase loads, where the indeterminacy is structurally richer because of the three-wire constraint.

\begin{figure*}[]
	\centering
	\includegraphics[width=\textwidth]{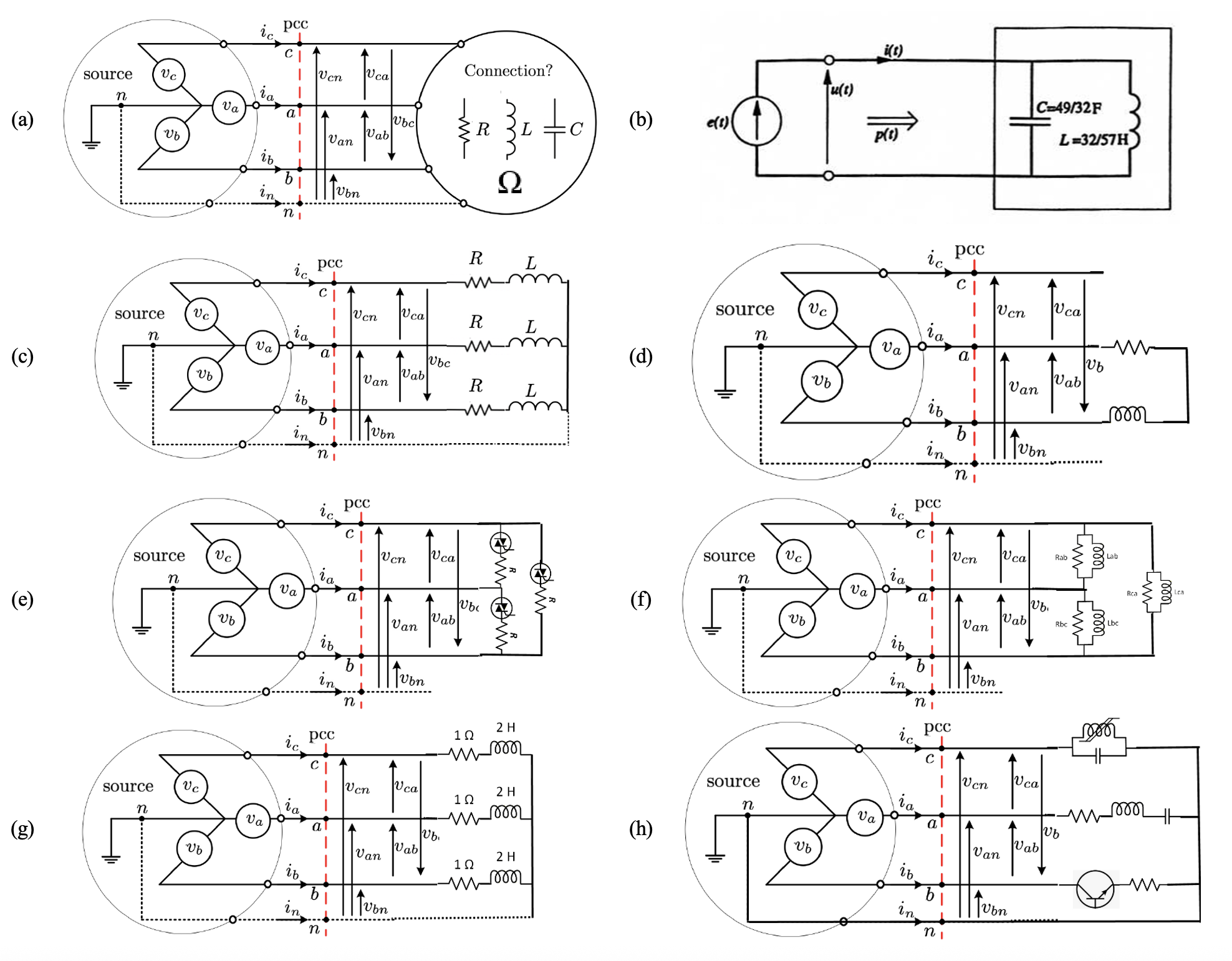}
	\caption{Figura 1}
	\label{fig:fig1}
\end{figure*}

\section{From Terminal Data to Branch Localization}
\label{sec:localization}

The localization framework requires an admissible three-phase lumped equivalent. Two families are used in the rest of the paper. The delta-parallel equivalent represents each line-to-line branch with parallel $G_{xy}$, $\Gamma_{xy}=1/L_{xy}$, and $C_{xy}$ elements,
\begin{equation}
i_{xy}=G_{xy}v_{xy}+\Gamma_{xy}\breve v_{xy}+C_{xy}v_{xy}',
\label{eq:delta_parallel}
\end{equation}
with line currents $i_a=i_{ab}-i_{ca}$, $i_b=i_{bc}-i_{ab}$, and $i_c=i_{ca}-i_{bc}$. The wye-series equivalent, referred to a real or virtual neutral $n$, represents each phase as series $R_x$, $L_x$, and $S_x=1/C_x$:
\begin{equation}
v_{xn}=R_x i_x+L_x i_x'+S_x\breve\imath_x.
\label{eq:wye_series}
\end{equation}
The virtual-neutral construction \vtwo{$v_{xn'}=\tfrac13(v_{xy}-v_{zx})$, with $(x,y,z)$ a cyclic permutation of $(a,b,c)$,} is needed in three-wire connections in which no physical neutral is measured \cite{blondel1893measurement,depenbrock1980active}.

For both families, branch parameters are recovered from terminal data by the time-domain identification methods established in \cite{montoyaTPWRD1,arrabal2025experimental} and extended to three-phase systems in \cite{montoya3phaseID}. The identification details are not reproduced here. The relevant property is that, given an excitation that is sufficiently rich in the time-domain sense, the parameters are determined uniquely up to numerical residuals, and the energy-localization terms in \eqref{eq:branch_decomposition} are then computed directly from the identified $\{R,L,C\}$ and the measured waveforms.


\section{Three-Phase Test Cases}
\label{sec:paradoxes}

This section reports six test cases. Each case is generated synthetically with known parameters in MATLAB, sampled at $20$~kHz or above over five fundamental periods at $50$~Hz. Reproducibility data, including waveforms, identified parameters, and balance residuals, are released as ancillary material.

\subsection{Balanced three-phase series $RL$ load: aggregated invariance hides per-phase exchange}
\label{subsec:balanced_RL}

A balanced three-phase $RL$ load \vtwo{(each phase a series $RL$ branch)} is supplied by a symmetric voltage triple $v_x=\sqrt2\,V\sin(\omega t-\theta_x)$ with $\theta_a=0$, $\theta_b=2\pi/3$, $\theta_c=4\pi/3$. With $|Z|=\sqrt{R^2+(\omega L)^2}$, $\theta=\arctan(\omega L/R)$, and $I=V/|Z|$, the line currents are $i_x=\sqrt2\,I\sin(\omega t-\theta_x-\theta)$; see Fig. \ref{fig:fig1}(c). The aggregated terminal power $p_{3\phi}=\sum_x v_x i_x$ equals $3VI\cos\theta$ and is constant in steady state. The aggregated magnetic energy
\begin{equation}
W_{L,3\phi}=\tfrac12 L\sum_x i_x^2=\tfrac32 LI^2
\label{eq:balanced_total_energy}
\end{equation}
is also constant, so $\sum_x W_{L,x}'=0$ at the three-phase boundary.

\begin{figure}[!t]
\centering
\includegraphics[width=\columnwidth]{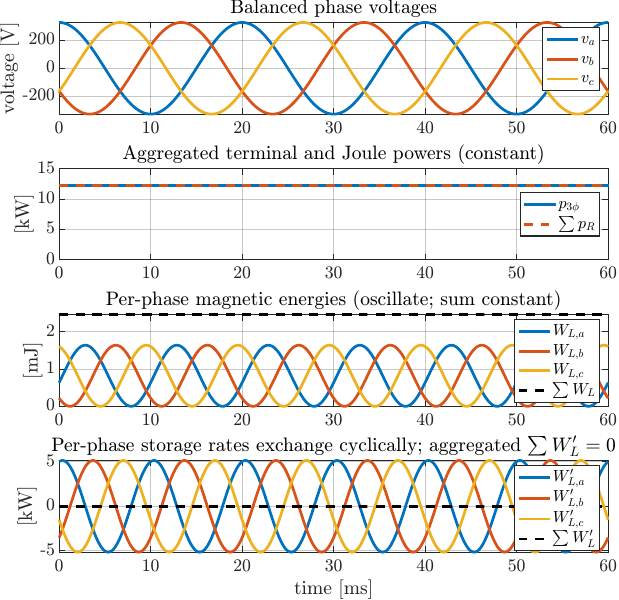}
\caption{\newblock{Balanced three-phase series $RL$ load. Top: terminal voltages. Second: aggregated terminal power and aggregated Joule dissipation, both constant in steady state. Third: per-phase magnetic energies $W_{L,x}(t)$ are time-varying; their sum is constant. Bottom: per-phase storage rates $W_{L,x}'$ exchange cyclically across phases while their sum vanishes.}}
\label{fig:balanced_RL}
\end{figure}

\newblock{The framework attaches an unambiguous physical reading. A reader who interprets the aggregated invariance as ``no reactive phenomenon present'' would be misled. Per-phase $W_{L,x}'\neq 0$, and physical magnetic energy genuinely oscillates inside each branch. The aggregated invariance is a geometric cancellation at the terminal surface, not a denial of internal exchange. Theorem~\ref{thm:localization} delivers the per-phase profile from the identified $L_x$ and the measured currents.}

\newblock{This case also exhibits a clean analytical separation between branch-level localization and the per-phase IEEE Std.~1459 instantaneous decomposition. The branch decomposition gives, for phase $a$,}
\begin{equation}
\begin{aligned}
p_{R,a}&=R i_a^2=P_{1\phi}\bigl[1-\cos(2\omega t-2\theta)\bigr],\\
W_{L,a}'&=L i_ai_a'=Q_{1\phi}\sin(2\omega t-2\theta),
\end{aligned}
\label{eq:branch_balanced_RL}
\end{equation}
with $P_{1\phi}=VI\cos\theta=R I^2$ and $Q_{1\phi}=VI\sin\theta=\omega L\,I^2$. \newblock{The IEEE Std.~1459 instantaneous active and reactive components per phase, taken from \cite{IEEE1459-2025}, are}
\begin{equation}
\begin{aligned}
p_{a,\rm act}^{\rm 1459}&=P\bigl[1-\cos 2\omega t\bigr],\\
p_{a,\rm rea}^{\rm 1459}&=-Q\sin 2\omega t,
\end{aligned}
\label{eq:1459_balanced_RL}
\end{equation}
\newblock{with $P=P_{1\phi}$ and $Q=Q_{1\phi}$ for this balanced case. Equations \eqref{eq:branch_balanced_RL} and \eqref{eq:1459_balanced_RL} have identical amplitudes but differ in phase by exactly $2\theta$; both pairs sum to the same instantaneous terminal product $v_a i_a$. Fig.~\ref{fig:balanced_RL_compare} reports the comparison for $V=120$~V, $R=1.1\,\Omega$, $L=1$~mH ($\theta\approx 16^\circ$, hence a $\sim 32^\circ$ relative phase shift between the two decompositions). The two decompositions agree on every aggregate (RMS magnitudes, time average) but disagree pointwise; only the branch decomposition associates an instant of $p_R$ with the actual instant of Joule dissipation in the branch resistor, and an instant of $W_L'$ with the actual instant of energy entering or leaving the branch inductor.}

\begin{figure}[!t]
\centering
\includegraphics[width=\columnwidth]{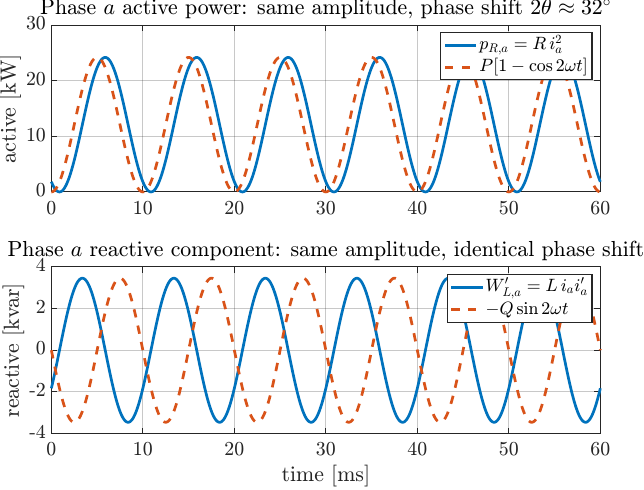}
\caption{\newblock{Balanced three-phase $RL$ load, phase $a$. Top: branch-level $p_{R,a}$ from \eqref{eq:branch_balanced_RL} versus the IEEE Std.~1459 instantaneous active per phase from \eqref{eq:1459_balanced_RL}. Bottom: branch-level $W_{L,a}'$ versus the standard's instantaneous reactive per phase. Same amplitudes, phase difference of $2\theta$ (see \cite{shepherdbook}). The two decompositions sum to identical $v_a i_a$ but localize the components differently in time.}}
\label{fig:balanced_RL_compare}
\end{figure}

\subsection{Single-branch resistive three-wire load: open phases, CPC ghost currents, and $p$--$q$ pseudo-reactive}
\label{subsec:resistive_unbalanced}

\newblock{Consider a three-wire load consisting of a single resistor $R_{ab}=0.65~\Omega$ connected between phases $a$ and $b$, with phases $b$--$c$ and $c$--$a$ open, energized by a balanced sinusoidal source. This case is a variant of the example used by de Leon and Cohen \cite{de2005discussion} to expose the ambiguity of waveform-based decompositions when only two phases conduct; see Fig. \ref{fig:fig1}(d).} The line currents satisfy $i_a=-i_b=v_{ab}/R_{ab}$ and $i_c\equiv 0$. The branch energy decomposition therefore returns
\begin{equation}
p_{\rm term}=p_{R,ab}(t)=R_{ab} i_{ab}^2,\quad
W_L'\equiv 0,\quad W_C'\equiv 0,
\label{eq:open_phase_balance}
\end{equation}
with all energy dissipated in $R_{ab}$ and zero physical storage anywhere.

\begin{figure}[!t]
\centering
\includegraphics[width=\columnwidth]{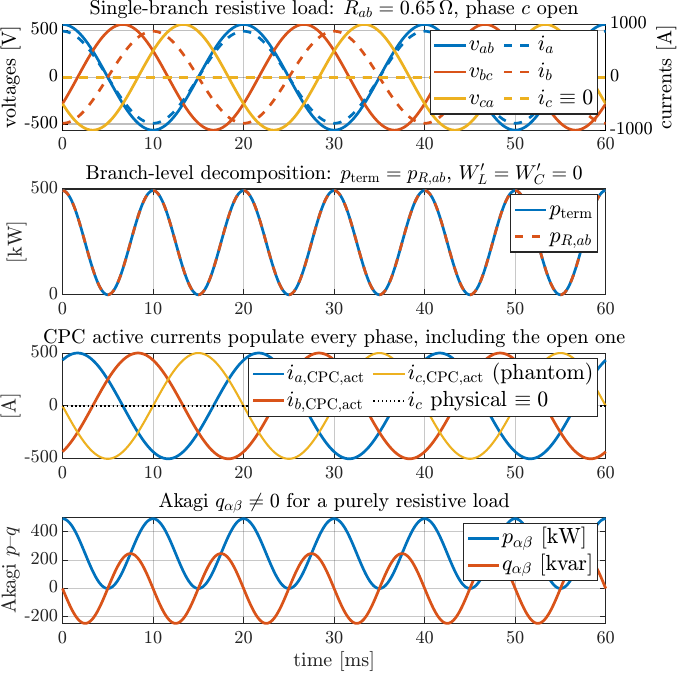}
\caption{\newblock{Single-branch resistive three-wire load: $R_{ab}=0.65~\Omega$ between phases $a$ and $b$, phase $c$ open. Top: line-to-line voltages and line currents; $i_c\equiv 0$. Second: terminal power coincides with $p_{R,ab}$ ; $W_L'=W_C'\equiv 0$. Third: Czarnecki CPC active currents \cite{czarnecki2008currents}, which are populated in all three phases including the open phase $c$. Bottom: Akagi instantaneous $p$--$q$ coordinates; $q_{\alpha\beta}\neq 0$ despite a purely resistive load.}}
\label{fig:resistive_unbalanced}
\end{figure}

\newblock{Classical decompositions do not return this answer directly (Fig.~\ref{fig:resistive_unbalanced}). Czarnecki's CPC active-current formula \cite{czarnecki2008currents} produces nonzero active currents in every phase, including phase $c$, whose physical line current is identically zero. The line currents are recovered as the algebraic sum of CPC components only after the introduction of an unbalanced current $i_u$ that cancels phase $c$, while the $a$ and $b$ active components are scaled to fit the equivalent conductance. The Akagi $p$--$q$ transformation likewise reports a nonzero $q_{\alpha\beta}$ for this purely resistive load, because the asymmetric currents in $\alpha\beta$ coordinates yield $v_\alpha i_\beta-v_\beta i_\alpha\neq 0$.}

\newblock{Czarnecki's own remark on the $p$--$q$ theory \cite{czarnecki1994comments}, ``any decomposition may lead to misinterpretation; in particular, zero can always be decomposed into nonzero terms,'' applies symmetrically to CPC in this scenario. The branch-level energy decomposition avoids the ghost components by demanding that every nonzero quantity be associated with a physical branch with identifiable parameters.}

\subsection{Three-phase nonlinear switched-resistive load: classical reactive powers as coordinate artifacts}
\label{subsec:nonlinear_switched}

A three-wire delta-connected symmetric resistor bank is supplied through three TRIACs in series with each delta branch and firing angle $\alpha$; see Fig. \ref{fig:fig1}(e). The load contains no inductor and no capacitor; all energy supplied by the source is converted to heat in the resistors. The line currents are nevertheless heavily distorted (Fig.~\ref{fig:nonlinear_switched}), with substantial low-order harmonics generated by the switching action.

Classical theories, applied to the same waveforms, return a rich landscape of nonactive components. Budeanu's distortion power $D$ is large because of the harmonic spread; Fryze's nonactive power $Q_F$ absorbs the orthogonal current to the supply; the Akagi $p$--$q$ pair has substantial $q_{\alpha\beta}(t)$; the Czarnecki CPC decomposition reports nonzero scattered and generated current components; and the CPT framework identifies a void-current term. Each of these quantities is mathematically valid and operationally useful for specific tasks, but none corresponds to a physical storage process inside the load.

\begin{figure}[!t]
\centering
\includegraphics[width=\columnwidth]{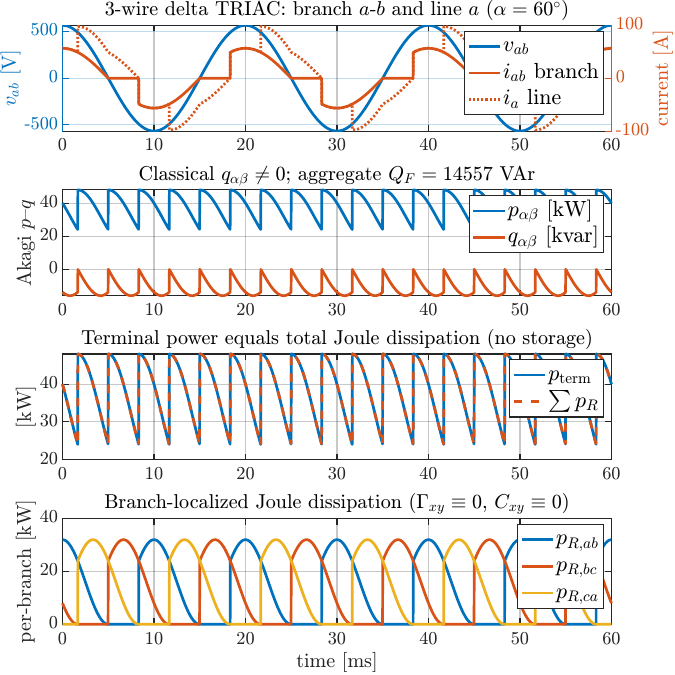}
\caption{\newblock{Three-wire delta-connected TRIAC-switched resistive load ($\alpha=60^\circ$). Top: line-to-line voltage $v_{ab}$, branch current $i_{ab}$ and line current $i_a=i_{ab}-i_{ca}$, showing strong waveform distortion. Second: Akagi $p$--$q$ instantaneous coordinates report substantial $q_{\alpha\beta}\neq 0$ and aggregate $Q_F\approx 14.6$~kVAr although the load contains no storage element. Third: terminal power $p_{\rm term}$ coincides with the sum of branch Joule dissipations, confirming $W_L'\equiv 0,W_C'\equiv 0$. Bottom: per-branch Joule dissipation $p_{R,xy}(t)$ — the entirety of $p_{\rm term}$ localizes to three switched conductances; the apparent classical nonactive components are coordinate-domain artifacts of representing the switching action by fixed-frequency waveform decompositions.}}
\label{fig:nonlinear_switched}
\end{figure}

\newblock{Identification with a parallel $RL$ topology returns $\Gamma_{xy}(t)\equiv 0$ within numerical residuals; identification with a parallel $RC$ topology returns $C_{xy}(t)\equiv 0$. The conductance $G_{xy}(t)$ is the same for both topologies and reproduces the switching action. The branch energy decomposition of Theorem~\ref{thm:localization} therefore assigns the entire $p_{\rm term}$ to Joule dissipation, with $W_L'\equiv 0$ and $W_C'\equiv 0$. The conclusion is sharp: the nonzero $D$, $Q_F$, $q_{\alpha\beta}$, scattered, generated, and void components reported by classical theories are coordinate-domain artifacts of representing a switched resistive process by fixed-frequency waveform decompositions; they describe the harmonic content of the current rather than any physical storage of energy.}

\newblock{Distinguishing operational coordinates from physical storage is precisely the purpose of branch-level localization, and switched resistive loads are the cleanest experimental setting in which the distinction can be made.}

\subsection{Three-wire delta-versus-wye topology indeterminacy}
\label{subsec:topology_indeterminacy}

This case is the constructive proof of Theorem~\ref{thm:indeterminacy}. We synthesize an unbalanced three-wire load with true parameters $(R_{ab},R_{bc},R_{ca})=(4,6,5)\,\Omega$ and $(L_{ab},L_{bc},L_{ca})=(10,12,11)$~mH in a delta-parallel $RL$ topology, energized by a balanced $400$~V$_{\rm rms}$ line-to-line source at $50$~Hz; see Fig. \ref{fig:fig1}(f). The simulation produces the four independent terminal signals $(v_{ab},v_{bc},i_a,i_b)$ shown in the top row of Fig.~\ref{fig:topology_indeterminacy}; from these data alone we identify two admissible equivalents:
\begin{enumerate}
\item the original delta-parallel $RL$ on the line-to-line edges, with branch parameters $(G,\Gamma)_{xy}=(1/R_{xy},1/L_{xy})$ recovered exactly (Fig.~\ref{fig:topology_indeterminacy}, left column);
\item a wye-series $RL$ referred to the virtual neutral $n'$ defined by \vtwo{$v_{xn'}=\tfrac13(v_{xy}-v_{zx})$ with $(x,y,z)$ cyclic in $(a,b,c)$}, fitted at the fundamental as $Z_x = V_{xn'}/I_x$ and reading $(R_a,R_b,R_c)=(0.56,0.47,0.59)\,\Omega$ and $(L_a,L_b,L_c)=(2.22,2.45,2.63)$~mH (right column).
\end{enumerate}
The wye-virtual parameters are visibly distinct from the truth, both in magnitude and in per-branch ratio.

\begin{figure}[!t]
\centering
\includegraphics[width=\columnwidth]{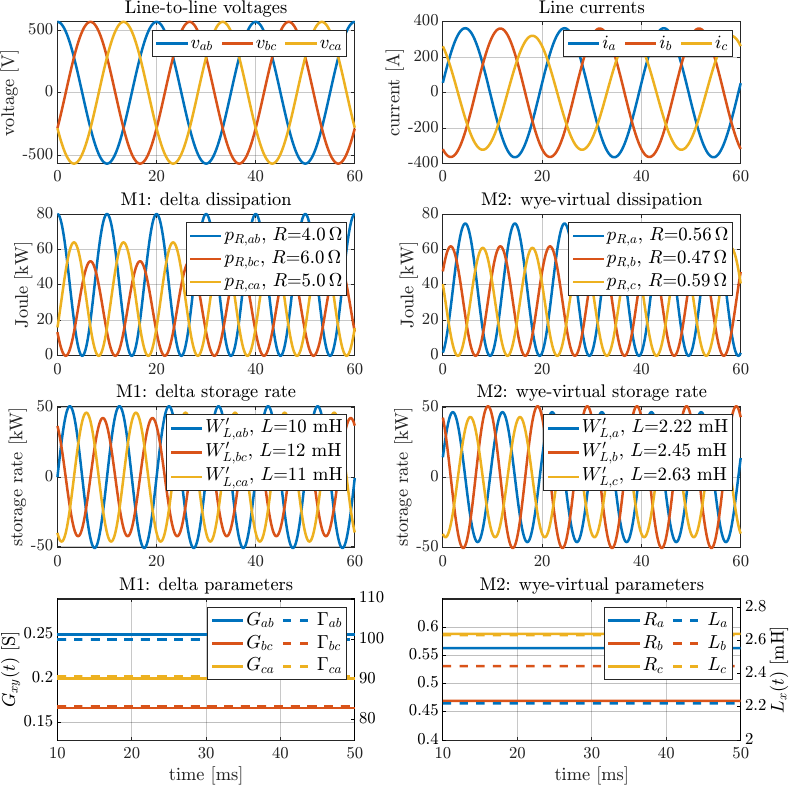}
\caption{\newblock{Three-wire topology indeterminacy. Top row: the four independent terminal signals $(v_{ab},v_{bc},v_{ca}; i_a,i_b,i_c)$ are common to both equivalents. Middle row: branch Joule dissipations $p_{R,xy}$ for the delta truth (left, with $R_{ab,bc,ca}=4,6,5\,\Omega$) and $p_{R,x}$ for the wye-virtual fit (right, with identified $R_{a,b,c}=0.56,0.47,0.59\,\Omega$); both share the same average $\bar P\approx 98.7$~kW but differ pointwise. Bottom row: branch storage rates $W_{L,xy}'$ for the delta description and $W_{L,x}'$ for the wye-virtual one, with $L_{ab,bc,ca}=10,12,11$~mH versus identified $L_{a,b,c}=2.22,2.45,2.63$~mH. Reconstruction residuals are at machine precision for both. Both localizations satisfy Theorem~\ref{thm:localization}; they differ because they postulate different internal topologies, in agreement with Theorem~\ref{thm:indeterminacy} and Corollary~\ref{cor:topology}.}}
\label{fig:topology_indeterminacy}
\end{figure}

\newblock{Despite the parameter mismatch, both equivalents reconstruct $p_{\rm term}$ to a relative residual of order $10^{-16}$ (Tellegen's identity, numerical), and they share the same average dissipation $\bar P\approx 98.7$~kW because the active power flow into the load is terminal observable. They differ in the time-resolved per-branch profiles. The middle row of Fig.~\ref{fig:topology_indeterminacy} shows that the delta branch dissipations $p_{R,xy}$ peak at $\sim 80$~kW per edge with envelopes set by the corresponding line-to-line voltages, whereas the wye-virtual $p_{R,x}$ peak at $\sim 75$~kW per phase with envelopes set by the line currents and a different inter-branch ordering. The bottom row shows the storage rates: delta branches exchange $\pm 50$~kW per edge through inductances of $10$--$12$~mH, while the wye-virtual phases reach the same $\pm 50$~kW range with substantially smaller inductances of $2.2$--$2.6$~mH; the two storage-rate envelopes differ pointwise in time.}

\newblock{The case is not a defect of the framework; it is the operational content of the topology-indeterminacy theorem. Both descriptions are admissible. The choice between them belongs to the modeler, who selects the topology that best matches the physical apparatus, the diagnostic question, or the available measurements.} A wye-virtual description is appropriate when the load is a star with accessible per-phase windings; a delta description is appropriate when the load is built from line-to-line elements. When neither is privileged, the modeler must report the choice explicitly and acknowledge that other localizations are possible.

\subsection{Three-phase fluctuating-phase load: vanishing fundamental reactive coordinates with manifest storage}
\label{subsec:fluctuating_phase}

A balanced four-wire extension of the fluctuating-phase example in \cite{jeltsema2015fluctuating} is built from the phase waveforms
\begin{equation}
\begin{aligned}
v_x&=V\sqrt2\cos(\omega_0t-\theta_x),\\
i_x&=I\sqrt2\cos[\omega_0t-\theta_x-\varphi(t)],
\end{aligned}
\label{eq:fluctuating_phase}
\end{equation}
where $\varphi(t)=-\gamma\sin\omega_1t$, $\theta_a=0$, $\theta_b=-2\pi/3$, and $\theta_c=2\pi/3$; see Fig. \ref{fig:fig1}(g). The numerical values follow the original single-phase example: $V=1$~V rms, $I=0.5$~A rms, $\omega_0=2\pi 50$~rad/s, $\omega_1=\omega_0/5$, $\gamma=\pi/2$, simulated with $f_s=200$~kHz over two modulation periods. The fundamental component of $i_x$ is in phase with $v_x$, so the classical fundamental reactive coordinate vanishes in each phase, and in the balanced three-phase lift the Akagi $p$--$q$ coordinate has zero mean ($\overline q_{\alpha\beta}\approx 0$) although it remains strongly oscillatory.
The same waveforms are identified pointwise with two admissible local models: parallel $RL$, $i_x=G_x(t)v_x+\Gamma_x(t)\breve v_x$ with $\Gamma_x=1/L_x$, and parallel $RC$, $i_x=G_x(t)v_x+C_x(t)v_x'$. The local fit at each sample uses two regressors $(v,\breve v)$ or $(v,v')$ together with the corresponding sample of $i$ and $i'$, so the recovered $G_x(t)$ exhibits a $2\omega_0$ ripple around the slow envelope $(I/V)\cos\varphi(t)$; this ripple is the signature of the instantaneous local fit and is reproduced by both topologies. Reference trajectories are obtained from the closed-form solution of the same local 2$\times$2 system; the reported finite-window estimates use the local windowed-identification procedure described in \cite{montoya3phaseID}.

\begin{figure}[!t]
\centering
\includegraphics[width=\columnwidth]{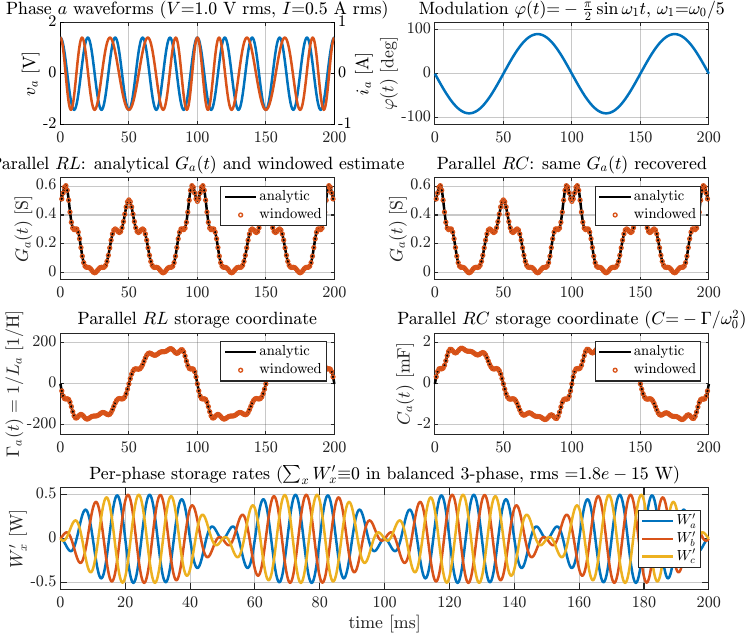}
\caption{\newblock{Three-phase fluctuating-phase load (balanced four-wire lift of the Jeltsema example) with $V=1$~V rms, $I=0.5$~A rms, $\omega_1=\omega_0/5$, $\gamma=\pi/2$. Top row: phase-$a$ voltage and current (left) and the modulation $\varphi(t)$ (right). Second row: analytical and windowed estimates of $G_a(t)$ for parallel $RL$ and parallel $RC$ descriptions; both recover the same conductance trajectory pointwise, including the $2\omega_0$ ripple superimposed on the $\cos\varphi(t)$ envelope. Third row: corresponding storage coordinates $\Gamma_a(t)=1/L_a(t)$ for $RL$ and $C_a(t)$ for $RC$, related by $C=-\Gamma/\omega_0^2$. Bottom row: per-phase storage rates $W_x'(t)$; their balanced 3-phase sum is numerically zero, so the aggregate three-phase terminal view hides the per-phase storage exchange. Plots show phase $a$ in the second and third rows; the windowed estimate uses windows of length $T_0/10$. Curves for phases $b$ and $c$ have the same shape with a $\pm 2\pi/3$ shift on the carrier.}}
\label{fig:fluctuating_phase}
\end{figure}

\newblock{Figure \ref{fig:fluctuating_phase} makes the fixed-coordinate limitation explicit. For the simulated values, $\overline p_{\alpha\beta}=0.708$~W and $\overline q_{\alpha\beta}\approx 0$, yet ${\rm rms}(q_{\alpha\beta})=1.21$~var and each phase has a storage-rate rms of $0.286$~W. The balanced aggregate $\sum_x W_x'$ is zero to numerical precision, so the three-phase terminal view hides the per-phase exchange that is manifest in the bottom panel of Fig.~\ref{fig:fluctuating_phase}. Branch localization remains valid because it operates on instantaneous waveforms and identified, possibly time-varying, parameters. The same dissipative trajectory $G_x(t)$ is recovered under both $RL$ and $RC$ storage assumptions; the storage parameter is the only variable that depends on the postulated topology, and it does so in a constrained way ($C=-\Gamma/\omega_0^2$), rather than being free to absorb arbitrary nonactive content. This cross-topology consistency is the operational signature that classical fixed-coordinate theories cannot deliver and that distinguishes a true storage process from a coordinate-domain artifact, extending the spirit of \cite{garcia2007power,jeltsema2015fluctuating} to three-phase systems.}

\subsection{Unbalanced four-wire nonlinear load: hysteretic, linear $RLC$, and switched branches in one section}
\label{subsec:fourwire_nonlinear}

\begin{figure}[!t]
\centering
\includegraphics[width=\columnwidth]{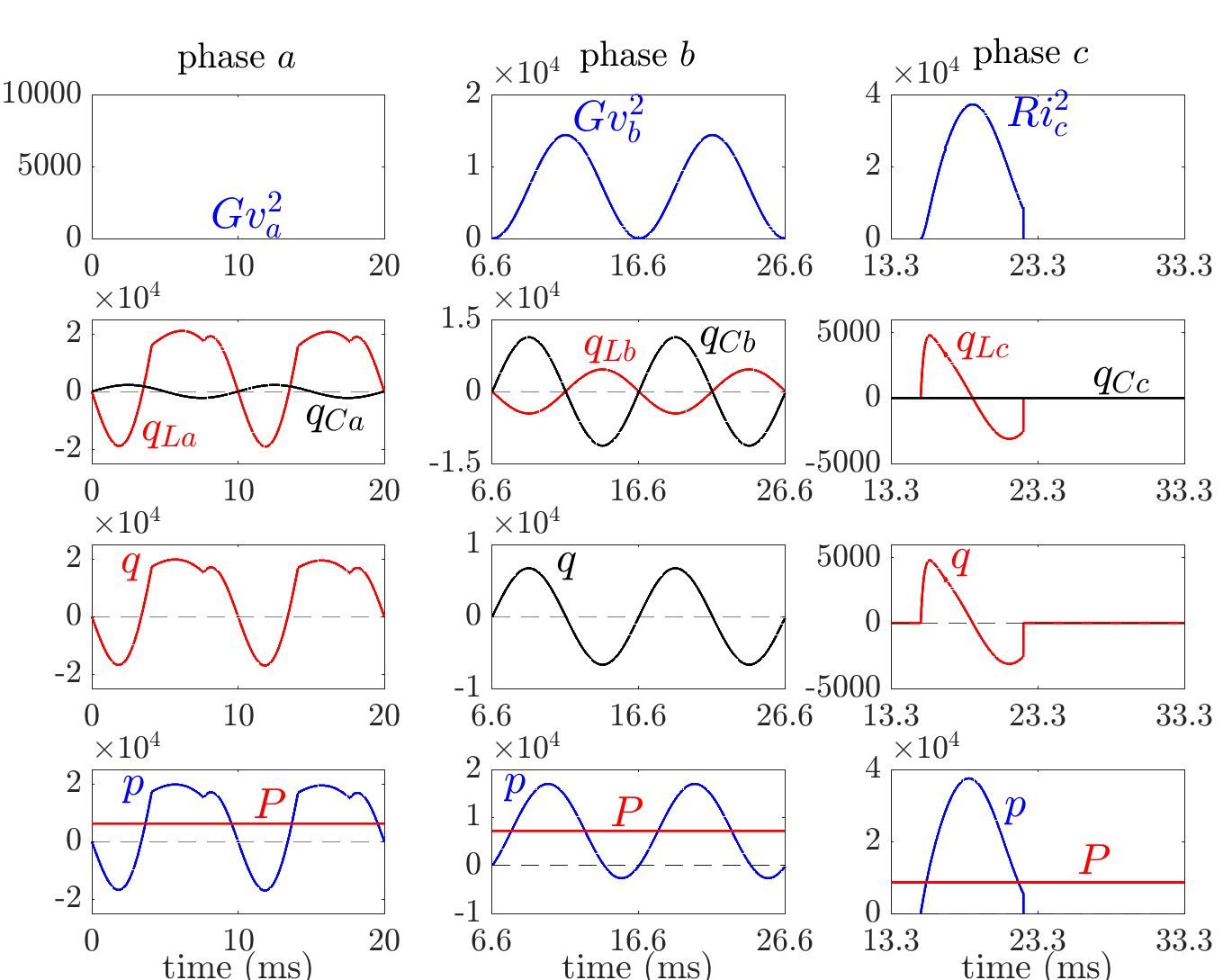}
\caption{\newblock{Powers for the three-phase circuit of the unbalanced four-wire nonlinear load. Each column corresponds to a phase. Top row: instantaneous active power consumed (energy-balance per branch). Second row: instantaneous reactive power associated with inductors ($q_L$) and capacitors ($q_C$). Third row: $q=q_L+q_C$. Bottom row: instantaneous total power $p=v\,i$ (in blue) and the average power $P$ (in red).}}
\label{fig:fourwire_nonlinear}
\end{figure}

\begin{figure}[!t]
\centering
\includegraphics[width=\columnwidth]{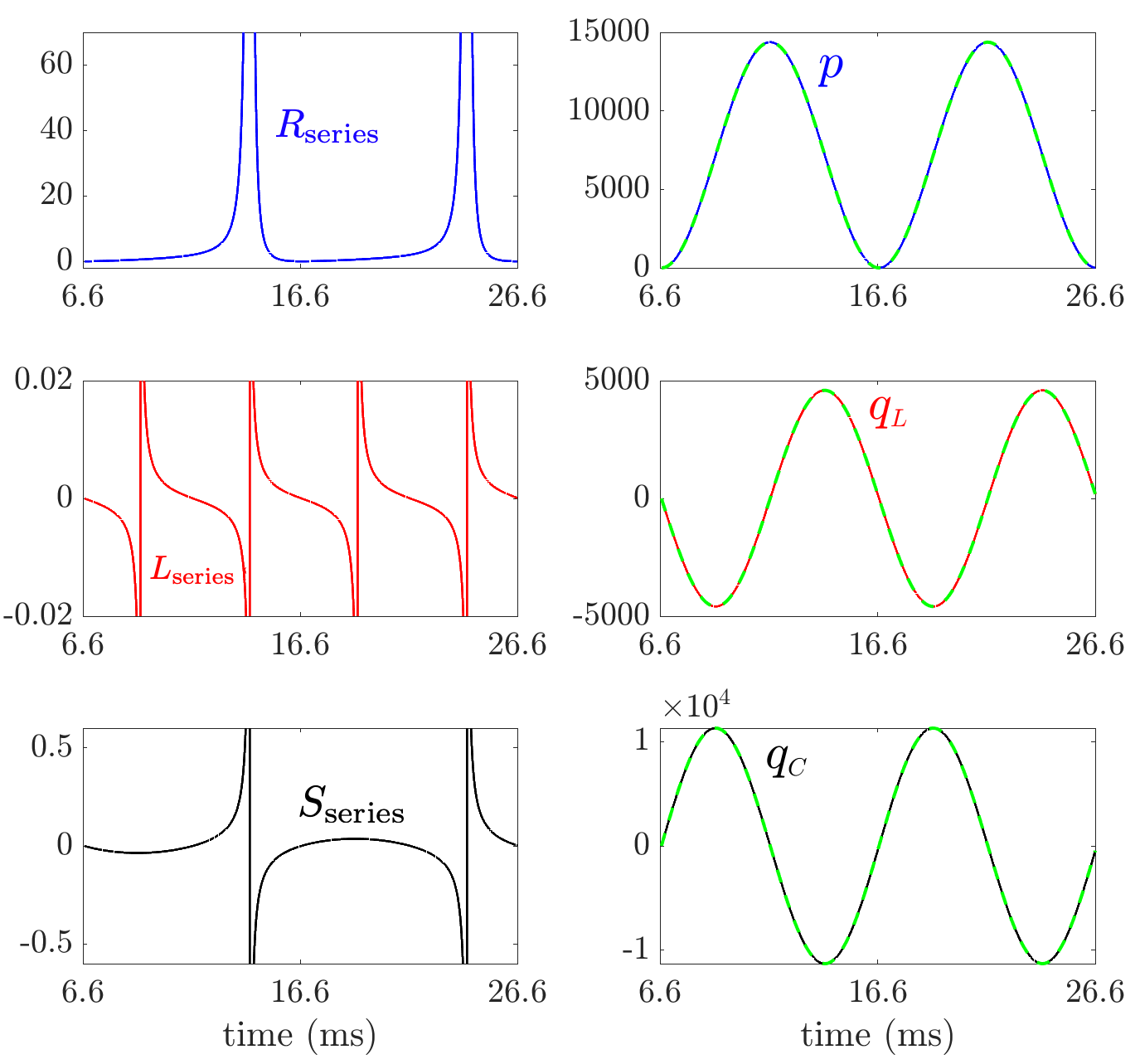}
\caption{\newblock{Energetic equivalent series circuit for the parallel circuit in phase $b$ of the four-wire nonlinear load. Left: parameters $R$, $L$, $S$. Right: active (consumed) and reactive powers.}}
\label{fig:duality}
\end{figure}

\newblock{The preceding cases share the three-wire constraint of Sections \ref{subsec:balanced_RL}--\ref{subsec:fluctuating_phase}. To document that branch-level localization is not specific to three-wire systems, this case takes a four-wire load with three structurally different per-phase branches: phase $a$ is a hysteretic inductor in parallel with a capacitor; phase $b$ is a linear parallel $RLC$ branch; phase $c$ is a purely resistive branch driven through a transistor whose conduction angle is $\alpha=\pi/3$ per half cycle; see Fig. \ref{fig:fig1}(h).} The supply is $v_x=120\sqrt2\sin(\omega t-\theta_x)+0.12\sqrt2\sin(3\omega t)$ at $f=50$~Hz, with the same harmonic content used in earlier identification work. The branch energy decomposition for the three phases is reported in Fig.~\ref{fig:fourwire_nonlinear}.

\newblock{Three observations follow. First, phase $a$ has no explicit resistor in the physical drawing, yet the identified parallel-branch admits a positive conductance $G_a$ that captures the hysteresis loss; the branch dissipation $p_{R,a}=G_a v_a^2$ accounts for the energy spent traversing the $B$--$H$ loop, with magnitude proportional to the loop area. Second, phase $b$ recovers exactly the linear $RLC$ pattern: $p_{R,b}=R_b i_{R,b}^2$, $W_{L,b}'$ and $W_{C,b}'$ have opposite signs over each subcycle, and their sum closes the per-phase balance. Third, phase $c$ has zero $p_R$ outside the conduction interval and a clean Joule trace during conduction; the framework treats the switching as a time-varying conductance and the closing residual remains at numerical precision pointwise.} The three branches obey three structurally different constitutive laws and the framework returns a single decomposition for each, with $\varepsilon\le 10^{-15}$ in every column.

\subsubsection{Series-parallel duality on phase $b$}
\label{subsubsec:duality_phaseb}

\newblock{Phase $b$ provides a clean numerical validation of the Generalized Energetic Duality (Theorem~\ref{thm:duality}). The branch is identified with the constant parallel parameters $(G_b,\Gamma_b,C_b)=(0.25\,\mathrm{S},100\,\mathrm{H}^{-1},100\,\mu\mathrm{F})$. The constructive map of Corollary~\ref{cor:series_parallel} returns the energetically equivalent series identification with time-varying parameters $R_s(t)=v^2 G_b/i^2$, $L_s(t)=v\breve v\,\Gamma_b/(i i')$, and $S_s(t)=v v'\,C_b/(i\breve\imath)$. Fig.~\ref{fig:duality} reports the result. The series parameters are genuinely time-varying, span both signs in some intervals (a mathematical signature of the energetic-dual representation rather than a physical realizability constraint), yet the branch energy components $p_R$, $W_L'$, and $W_C'$ are identical pointwise to those of the constant-parameter parallel description, with maximum difference of order $10^{-8}$~W. This figure is the multi-frequency, time-domain analogue of the classical series-parallel equivalence used in LTI sinusoidal analysis: the equivalence is preserved, but the time-invariance of the equivalent parameters is not.}

\section{Comparative Analysis}
\label{sec:comparison}

Table~\ref{tab:comparative} summarizes how IEEE Std.~1459, the Akagi instantaneous $p$--$q$ formulation, the Czarnecki CPC decomposition, the Fryze-Buchholz-Depenbrock framework, and the present branch-level localization respond to the six test cases. The numerical entries are computed from the same synthetic waveforms by independent implementations of each theory. The qualitative columns indicate whether the theory localizes physical storage; the framework of this paper is the only one that returns branch-level $p_R$, $W_L'$, and $W_C'$ profiles, precisely because it requires a topology and an identified equivalent rather than only terminal coordinates.

\begin{table*}[!t]
\centering
\caption{\newblock{Comparative response of classical theories and branch-level localization on the six test cases.}}
\label{tab:comparative}
\scriptsize
\setlength{\tabcolsep}{3pt}
\begin{tabular}{@{}p{0.10\textwidth}p{0.17\textwidth}p{0.15\textwidth}p{0.13\textwidth}p{0.13\textwidth}p{0.21\textwidth}@{}}
\toprule
Case & IEEE Std.~1459 & Akagi $p$--$q$ & Czarnecki CPC & FBD & Branch-Level Localization\\
\midrule
Balanced $RL$ (\S\ref{subsec:balanced_RL}) & Per-phase $p_a^{\rm 1459}\!=\!P[1\!-\!\cos 2\omega t]$; phase-shifted $2\theta$ from branch form & $\bar p$ constant, $\bar q\neq 0$; per-phase profile not exposed & Reactive current globally; per-phase storage not assigned & Active and nonactive currents globally; no per-phase split & Per-phase $W_{L,x}'(t)$ recovered; analytical $2\theta$ phase shift exposed\\
Resistive open phase (\S\ref{subsec:resistive_unbalanced}) & $S>P$ from one-branch asymmetry & $q_{\alpha\beta}\neq 0$ for purely resistive load & Active currents in all 3 phases including the open one & Nonactive current present & $W_L'=W_C'\equiv 0$; entire $p_{\rm term}$ as Joule in $R_{ab}$\\
Nonlinear switched (\S\ref{subsec:nonlinear_switched}) & Large $D$, $Q_F$ from harmonic spread & $q_{\alpha\beta}\neq 0$ despite no storage element & Scattered, generated currents nonzero & Void/nonactive current populates & $W_L'=W_C'\equiv 0$; time-varying $G(t)$ captures switching\\
Topology indeterminacy (\S\ref{subsec:topology_indeterminacy}) & Single $S$ regardless of topology & Single $p$--$q$ pair regardless of topology & Single decomposition regardless of topology & Single nonactive current regardless of topology & Two distinct localizations; both admissible (Cor.~\ref{cor:topology}); series-parallel duality (Cor.~\ref{cor:series_parallel})\\
Fluctuating phase (\S\ref{subsec:fluctuating_phase}) & $Q_1\!\to\!0$; storage misclassified as distortion & $\bar q_{\alpha\beta}\!\to\!0$; storage missed & Reactive current of fundamental small & Nonactive current small in fundamental & Time-varying $W_L'$ or $W_C'$ recovered; cross-topology $G_x(t)$ check\\
4-wire nonlinear (\S\ref{subsec:fourwire_nonlinear}) & Per-phase $S$, no decomposition of hysteresis loss & $p$--$q$ aggregates across phases; switching not localized & Decomposes per phase but assigns no storage to hysteresis & Nonactive current populates without locating loss & Per-phase $p_R$, $W_L'$, $W_C'$ identified; hysteresis loss localized to phase $a$\\
\bottomrule
\end{tabular}
\end{table*}

\newblock{The reading of Table~\ref{tab:comparative} is positive, not adversarial. Classical theories do not contradict themselves on the test cases; they answer the question they were designed for. They report aggregated rating quantities, projection onto fixed coordinates, or current decompositions that admit specific compensator designs. The branch-level localization answers a different question: where, inside the load, is energy dissipated or stored at any instant? The two answers are compatible and serve complementary needs.}

Three regimes can be identified on the test set. In the first regime, balanced sinusoidal LTI loads, all theories return numerically equivalent answers up to definition-dependent constants, and the branch decomposition coincides with the standard $P$, $Q_1$ pair scaled by the topology. In the second regime, including the balanced and unbalanced resistive cases, classical theories disagree on the labeling of nonactive components but agree numerically on aggregated $S$ or its analogue; the branch decomposition is consistent with all of them once the labels are resolved against the identified topology. In the third regime, including the topology-indeterminacy and fluctuating-phase cases, classical theories return well-defined coordinate-domain numbers that do not, by themselves, decide the physical question of storage location; only an identified equivalent does.
\section{Discussion}
\label{sec:discussion}
\subsection{Scope and limitations}
The framework operates at the lumped-circuit level. Field-theoretic interpretations of polyphase energy flow~\cite{emanuel2004poynting,calamaro2015review} address questions that the present analysis does not, including the physical localization of energy on conductor surfaces and the role of return paths. Spatial-vector formulations~\cite{ferrero1991space} project waveforms onto a fixed two-dimensional plane and provide compact diagnostic indices that the branch decomposition complements rather than competes with.

The localization is unique only once a topology is fixed. When several admissible topologies coexist, as in three-wire systems without measured neutral, the choice is a modeling act and must be reported. Theorem~\ref{thm:indeterminacy} and Corollary~\ref{cor:topology} are the formal acknowledgment of this constraint. Practical guidance: select the topology that matches the physical apparatus when known; select the topology that yields the lowest energy residual $\rho_E$ otherwise; report the choice and, if possible, both alternative localizations.

The synthetic test cases isolate the physics from measurement-chain effects. Real records require uncertainty handling, synchronization control, and derivative-or-primitive reconstruction; these aspects are addressed in the identification literature~\cite{montoyaTPWRD1,arrabal2025experimental}  that supplies the equivalents used here. When the analysis window does not contain enough waveform diversity to resolve all branch parameters, the localization framework is honest about it: rank, conditioning, and residual diagnostics flag the case rather than producing a misleading branch decomposition.
\subsection{Operational implications}
Three implications stand out. First, compensator design benefits from operating on identified storage rather than on aggregated reactive figures, particularly when the operational objective includes phase-loading equalization in addition to feeder-loss minimization. Second, power-quality classification can move from waveform-based labels (harmonic, distortion, asymmetry) to mechanism-based labels (Joule, magnetic exchange, electric exchange, time-varying storage). Third, fault detection benefits from the appearance of energy terms inconsistent with the nominal topology, because such mismatches are direct signatures of structural change.
\subsection{Position relative to classical theories}
\newblock{The position adopted in this paper is conservative on definitions and operational on interpretation. No classical theory is declared invalid. IEEE Std.~1459 remains the legal and metering reference. CPC remains a powerful diagnostic and compensation tool. The instantaneous $p$--$q$ formulation remains the foundation of converter control. Fryze-Buchholz-Depenbrock remains the natural framework for non-active currents in multi-conductor circuits. Each answers a question that the branch-level localization does not. The localization framework, in turn, answers a question they do not: where, instant by instant, is the energy dissipated or stored?}
\section{Conclusions}
\label{sec:conclusion}
This paper develops a branch-level energy-localization framework for three-phase loads. The instantaneous terminal power of an admissible lumped equivalent is decomposed uniquely into Joule dissipation and the time rate of stored magnetic and electric energy, branch by branch (Theorem~\ref{thm:localization}). When multiple admissible topologies reproduce the same terminal data, the localizations differ; the topology choice is a modeling decision (Theorem~\ref{thm:indeterminacy}, Corollary~\ref{cor:topology}). A Generalized Energetic Duality Theorem (Theorem~\ref{thm:duality}) shows that, for any branch terminal pair, every admissible topology family can be parameterized to reproduce the same branch energy decomposition: classical Norton-Thevenin, series-parallel, and $L\leftrightarrow C$ dualities are restrictions of this principle to LTI sinusoidal regimes, where the time-varying parameter trajectories degenerate to constants. Six test cases illustrate the framework on situations that classical theories describe in incomplete or inconsistent ways: aggregated invariance hiding per-phase exchange and the analytical phase-shift between branch-level and IEEE Std.~1459 instantaneous decompositions; an open-phase resistive load where CPC produces ghost active currents in every phase and $p$--$q$ reports pseudo-reactive components for a purely resistive circuit; switched-resistive nonlinear loads in which classical distortion, scattered, void, and instantaneous reactive components are coordinate-domain artifacts of a purely dissipative process; three-wire topology indeterminacy; fluctuating-phase loads invisible to fixed-coordinate projections; and a four-wire nonlinear load with hysteretic, linear, and switched branches in a single section, where hysteresis loss is localized to its branch via an effective time-varying conductance. Future work includes experimental implementation with full measurement-chain handling, extension to time-varying nonlinear elements with explicit nonlinear constitutive relations, and integration of the localization with port-Hamiltonian and field-theoretic analyses.

\bibliographystyle{IEEEtran}
\bibliography{rebuild_references}

\end{document}